\documentclass[journal]{IEEEtran}
\usepackage{amsmath,amssymb,amsthm,graphicx,cite}
\usepackage{bbm}
\usepackage{subcaption}
\newtheorem{theorem}{Theorem}
\newtheorem{proposition}{Proposition}

\newtheorem{remark}{Remark}
\newtheorem{assumption}{Assumption}

\begin{document}

\title{Battery-Aware Rate-Splitting Multiple Access for Solar-Powered Cell-Free LEO Satellite Networks}
\author{Wali Ullah Khan, Muhammad Adil\thanks{Wali Ullah Khan is with the SnT, University of Luxembourg, Luxembourg (email: waliullahkhan30@gmail.com). Muhammad Adil is with the Department of Electronics Engineering, University of Rome Tor Vergata, 00133 Rome, Italy (email: muhammad.adil@uniroma2.it).

}}

\maketitle

\begin{abstract}
Cell-free low Earth orbit (LEO) satellite downlinks improve coverage and macro-diversity, but intermittent solar harvesting and finite onboard batteries can make communication-only resource allocation energy-aggressive. We develop battery-aware one-layer rate-splitting multiple access (BA-RSMA) for a fixed-cluster solar-powered cell-free LEO downlink. A perturbed physical-battery Lyapunov queue couples common/private power allocation to stored energy, while robust energy causality protects against bounded harvesting uncertainty and allows harvest curtailment at battery saturation. Under a fixed-precoder scalar effective-channel model, dual and quadratic transforms yield convex fixed-auxiliary resource-allocation subproblems and monotonic alternating updates. Numerical cross-validation against a direct nonlinear-programming reference gives a maximum relative objective gap below $2.8\times10^{-7}$ over 12 representative slot--state instances. In a 20-seed battery-stressed experiment, BA-RSMA achieves 1.631~Mbit/J and reduces mean near-depletion from 24.84\% under Myopic-RSMA to 6.82\%, at only a 1.7\% EE penalty.
\end{abstract}


\section{Introduction}

User-centric cell-free non-terrestrial networks (NTNs) enable each user equipment (UE) to be jointly served by a geometry-dependent cluster of low Earth orbit (LEO) satellites, thereby improving macro-diversity and mobility robustness relative to satellite-centric multi-beam operation \cite{kim2025cellfree}. Existing satellite resource-allocation designs commonly impose fixed transmit-power budgets, although a LEO platform harvests solar energy only during illuminated orbital intervals and supplies its communication and platform loads from a finite battery \cite{wertz1999space}. In parallel, rate-splitting multiple access (RSMA) provides a flexible interference-management mechanism by superposing a common stream, decoded by all users, and user-specific private streams \cite{khan2023rsma,11397682}. Lyapunov drift-plus-penalty control offers a natural means of coupling such per-slot communication decisions to long-term battery operation without requiring noncausal harvesting information \cite{neely2010stochastic,huang2013energyharvesting}.

Recent cell-free LEO studies have addressed stochastic-geometry performance~\cite{li2026downlink}, multi-satellite macro-diversity~\cite{dandrea2025macrodiversity}, and beamforming/resource allocation~\cite{gao2026wsr}. Energy-aware satellite designs have considered satellite--UAV and hybrid satellite--terrestrial cell-free architectures~\cite{tran2024satelliteuav,thu2026loadbalancing}, Lyapunov-based stochastic resource allocation in integrated satellite--terrestrial networks~\cite{koutsioumpa2025lyapunov}, and energy-constrained LEO edge computing under long-term energy limits~\cite{cheng2025energy}. RSMA has also been investigated in integrated satellite--terrestrial cell-free systems~\cite{zhang2025rsma} and energy-efficient hybrid satellite--terrestrial networks with energy harvesting~\cite{mittal2026resource}. However, these works do not explicitly couple one-layer RSMA common/private power allocation to the physical finite-battery state of the solar-powered serving LEO satellites under bounded harvesting uncertainty. An energy-efficient transmission decision can therefore remain aggressive before or during eclipse intervals.

Motivated by this gap, we develop BA-RSMA for a fixed-cluster solar-powered cell-free LEO downlink. The controller couples common/private RSMA power allocation to the measured battery reserve. Satellite association and beam directions are fixed so that coherent multi-satellite transmission reduces to scalar power allocation over an effective channel, making communication load affine in aggregate transmit power.

The contributions are threefold. First, we formulate battery-aware one-layer RSMA with finite batteries, QoS constraints, and robust energy causality under bounded harvesting error. Second, a perturbed-battery Lyapunov controller is combined with dual and quadratic transforms, yielding convex fixed-auxiliary subproblems and a monotonic alternating algorithm with stationary-point convergence guarantees. Third, we quantify the EE--reserve trade-off and compare BA-RSMA with an identical Myopic-RSMA ablation ($Z_i=0$), reduced-order BA-RSMA, NOMA, and OMA, while cross-validating the proposed alternating solver against a direct nonlinear-programming reference.

\textit{Organization and notation:} Section~II presents the system and battery models; Section~III formulates the battery-aware EE problem and develops the proposed alternating solution; Section~IV reports the numerical results; and Section~V concludes the paper. Scalars are italic, vectors are bold lowercase, and sets are calligraphic uppercase. $(\cdot)^{\mathsf H}$, $\|\cdot\|$, $\mathbb{E}[\cdot]$, and $[x]^+\triangleq\max\{x,0\}$ denote Hermitian transpose, Euclidean norm, expectation, and the positive-part operator, respectively.

\begin{figure}[t]
\centering
\includegraphics[width=0.8\columnwidth]{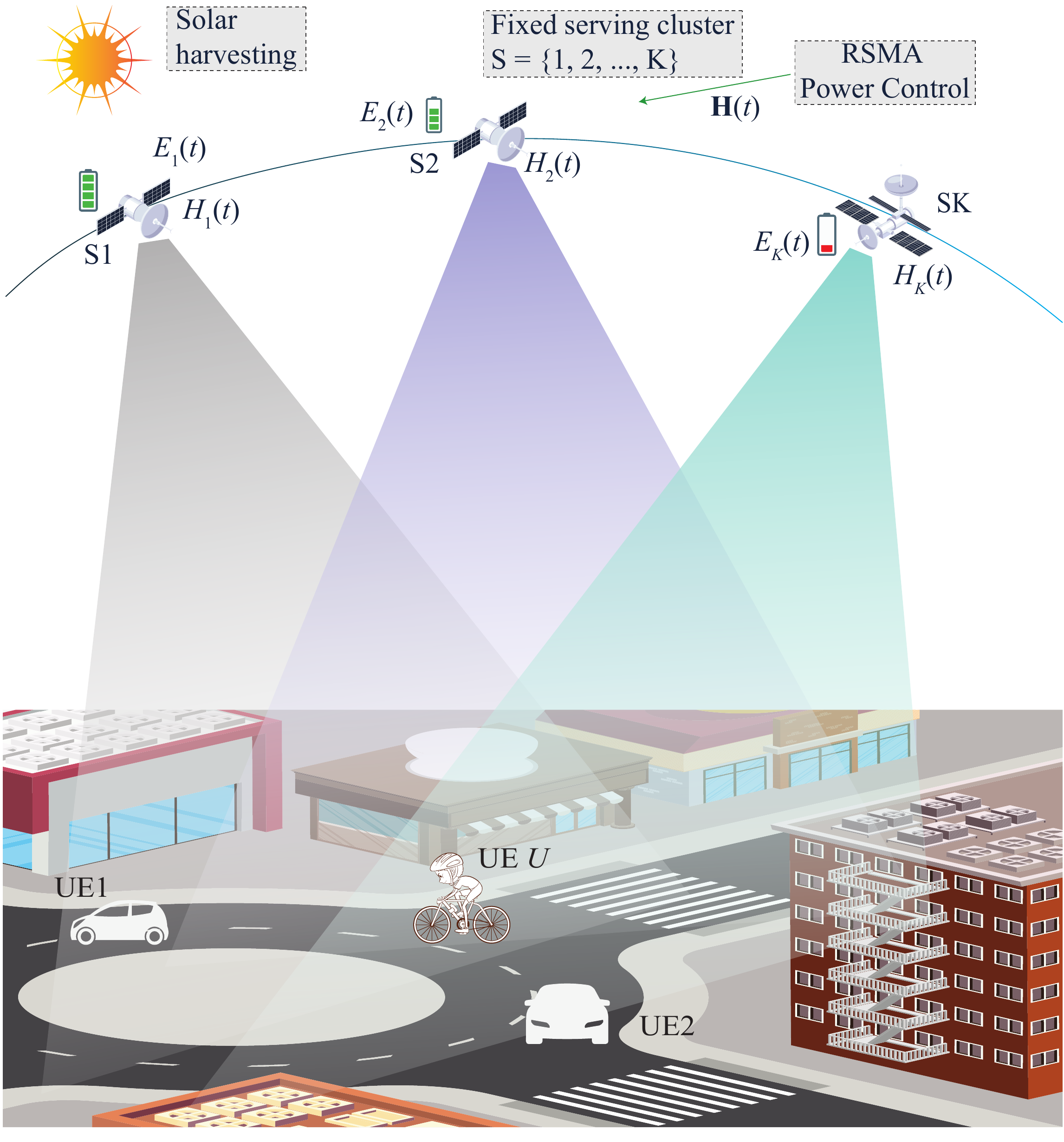}
\caption{System model of battery-aware cell-free LEO network.}
\label{fig:SM}
\end{figure}

\section{System Model}

\subsection{Network and One-Layer RSMA Model}
We consider a downlink cell-free LEO network in which a fixed
cluster of $K$ satellites, indexed by
$\mathcal{S}=\{1,\ldots,K\}$, jointly serves $U$ single-antenna
user equipments (UEs), indexed by
$\mathcal{J}=\{1,\ldots,U\}$. As illustrated in Fig.~\ref{fig:SM}, the satellites experience alternating
sunlit and eclipse conditions and jointly transmit to the served
UEs over the same time--frequency resource. Cluster membership is fixed over the control horizon; no clustering variable appears in the per-slot optimization. With fixed beam directions, coherent multi-satellite transmission is represented by
\begin{equation}
h_j(t) \triangleq \Big|\textstyle\sum_{i\in\mathcal S}
\sqrt{\xi_i\beta_{i,j}(t)\rho_{i,j}(t)}\,
\mathbf a_{i,j}^{\mathsf H}(t)\mathbf v_{i,j}(t)\Big|^2,
\label{eq:effective_gain}
\end{equation}
where $\beta_{i,j}$ is the large-scale power gain, $\kappa_{i,j}$ the Rician $K$-factor, $\rho_{i,j}=\kappa_{i,j}/(1+\kappa_{i,j})$ the LOS power fraction, $\mathbf a_{i,j}$ the normalized steering vector, and $\mathbf v_{i,j}$ a fixed unit-norm precoder. Timing, Doppler, and carrier-phase offsets are assumed compensated sufficiently for coherent transmission within each slot. The desired and interfering streams at UE $j$ therefore share $h_j(t)$ in \eqref{eq:common_sinr}--\eqref{eq:private_sinr}; this is a cluster-level power-control abstraction rather than a general multiuser beamforming model.

\textit{One-layer RSMA:} 
Under the resulting scalar channel model, the cluster employs
one-layer RSMA, transmitting one common stream decoded by
every UE via SIC together with $U$ private streams. Let
$p_c(t)\geq0$ and $p_j(t)\geq0$ denote the common- and
private-stream transmit powers, respectively. The common- and private-stream SINRs are
\begin{align}
\Gamma_{c,j}(t)
&=
\frac{p_c(t)h_j(t)}
{\sum_{k}p_k(t)h_j(t)+\sigma^2},
\label{eq:common_sinr}\\
\Gamma_{p,j}(t)
&=
\frac{p_j(t)h_j(t)}
{\sum_{k\neq j}p_k(t)h_j(t)+\sigma^2}.
\label{eq:private_sinr}
\end{align}
Since the common message must be decodable by every UE, its rate is capped by the weakest user and split into non-negative per-user portions $t_j(t)$:
\begin{equation}
\begin{aligned}
\sum_j t_j(t)
&\le
\log_2\!\left(1+\min_k\Gamma_{c,k}(t)\right),\\
R_j(t)
&=
t_j(t)+\log_2\!\left(1+\Gamma_{p,j}(t)\right).
\end{aligned}
\label{eq:rsma_rates}
\end{equation}
This single shared common stream constitutes the one-layer RSMA structure considered throughout the paper.

\subsection{Energy Model}
Each satellite $i\in\mathcal{S}$ draws a fixed, predetermined share $\xi_i\geq0$ ($\sum_i\xi_i=1$) of the cluster's total radiated power. Accordingly, the coherent amplitude contributed by satellite $i$ carries the factor $\sqrt{\xi_i}$ in the effective gain above, while its RF communication load remains affine in the aggregate stream powers,
\begin{equation}
e_i(t) = \frac{\xi_i}{\eta_i}\Big(p_c(t)+\textstyle\sum_j p_j(t)\Big),
\end{equation}
with $\eta_i$ the amplifier efficiency. The fixed weight $\xi_i$ is computed offline from nominal service-region geometry and long-term statistics. With $j\sim\mathcal D$ denoting the design distribution of served terminals,
\begin{equation}
\xi_i = \frac{\mathbb E_{j\sim\mathcal D}[\beta_{i,j}\rho_{i,j}|\mathbf a_{i,j}^{\mathsf H}\mathbf v_{i,j}|^2]}
{\sum_{\ell\in\mathcal S}\mathbb E_{j\sim\mathcal D}[\beta_{\ell,j}\rho_{\ell,j}|\mathbf a_{\ell,j}^{\mathsf H}\mathbf v_{\ell,j}|^2]}.
\label{eq:power_share}
\end{equation}
Thus \eqref{eq:power_share} is causal and should be viewed as a fixed power-sharing proxy, not an exact decomposition of the coherent gain because $|\sum_i a_i|^2$ contains cross terms. Optimizing $\xi_i$ per slot would reintroduce joint satellite-level power allocation and is outside the present scalar scope.

Let $H_i(t)\geq0$ denote the harvestable solar power at satellite $i$ during slot $t$ ($H_i(t)=0$ in eclipse), $\theta_i\in(0,1)$ the charging efficiency, $\Delta$ the slot duration, $P_i^{\mathrm{bus}}(t)$ the non-communication (attitude-control, telemetry, and housekeeping) load, and $P_i^{\mathrm{fix}}$ the fixed communication/circuit load. The battery-supplied load is therefore $P_i^{\mathrm{load}}(t)=P_i^{\mathrm{bus}}(t)+P_i^{\mathrm{fix}}+e_i(t)$. Rate-dependent onboard forwarding energy is neglected; the model retains the RF power, fixed communication circuitry, and dominant platform bus load. To avoid the unphysical requirement that a satellite must transmit solely to prevent battery overflow, we allow excess harvested energy to be curtailed when the battery is full. The physical battery recursion is therefore
\begin{equation}
E_i(t+1)=\min\!\left\{E_i^{\max},\;E_i(t)+\Delta\big(\theta_iH_i(t)-P_i^{\mathrm{load}}(t)\big)\right\}.
\label{eq:battrec}
\end{equation}
The saturation in \eqref{eq:battrec} represents standard harvest curtailment/spillage and is not a post-hoc correction of battery depletion. With $H_i$ and $P_i^{\mathrm{load}}$ in watts, $\Delta$ must use the corresponding time base for the stored energy; in the numerical study, battery energy is represented in Wh and the slot duration is converted to hours when multiplying power in watts.

\begin{assumption}\label{as:causal}
At the start of slot $t$, the controller knows $E_i(t)$, the current scalar effective channel gains $\{h_j(t)\}_{j\in\mathcal J}$, and a causal estimate $\hat H_i(t)$ of the harvestable power, with $|\hat H_i(t)-H_i(t)|\leq\delta_{H,i}$ for a known bound $\delta_{H,i}\geq0$. The true update \eqref{eq:battrec} uses the realized $H_i(t)$; only the controller's action uses $\hat H_i(t)$.
\end{assumption}

To robustly prevent depletion, the implemented load must satisfy
\begin{equation}
\Delta P_i^{\mathrm{load}}(t)
\leq E_i(t)-E_i^{\min}
+\Delta\theta_i\big[\hat H_i(t)-\delta_{H,i}\big]^+.
\label{eq:causal1}
\end{equation}
Equivalently, the communication-power component obeys
\begin{align}
e_i(t)\leq e_i^{\max}(t)\triangleq{}&
\frac{E_i(t)-E_i^{\min}}{\Delta}
+\theta_i\big[\hat H_i(t)-\delta_{H,i}\big]^+\notag\\
&-P_i^{\mathrm{bus}}(t)-P_i^{\mathrm{fix}}.
\label{eq:emax}
\end{align}
where slots for which $e_i^{\max}(t)<0$ are infeasible even at zero communication power; throughout, the combined bus and fixed communication loads are assumed to be energy-feasible even at zero radiated power.

\begin{proposition}\label{prop:safety}
If $E_i(0)\in[E_i^{\min},E_i^{\max}]$, Assumption~\ref{as:causal} holds, and every implemented action satisfies \eqref{eq:causal1}, then $E_i(t)\in[E_i^{\min},E_i^{\max}]$ for all $t$ for every harvest realization satisfying the stated error bound.
\end{proposition}
\begin{proof}
Suppose $E_i(t)\in[E_i^{\min},E_i^{\max}]$. Since $H_i(t)\geq\hat H_i(t)-\delta_{H,i}$ and $H_i(t)\geq0$, we have $\theta_iH_i(t)\geq\theta_i[\hat H_i(t)-\delta_{H,i}]^+$. Hence \eqref{eq:causal1} gives
\[
E_i(t)+\Delta\big(\theta_iH_i(t)-P_i^{\mathrm{load}}(t)\big)\geq E_i^{\min}.
\]
Taking the minimum with $E_i^{\max}$ in \eqref{eq:battrec} cannot reduce the state below $E_i^{\min}$ and directly enforces the upper bound. Induction from the feasible initial state completes the proof.
\end{proof}

\section{Problem Formulation and Proposed Solution}

\subsection{Communication-Side Energy Efficiency}
We use the communication-side effective EE
\begin{equation}
\Xi(t)\triangleq \frac{W\sum_jR_j(t)}{P_0+\sum_i e_i(t)},
\label{eq:ee}
\end{equation}
where $P_0=\sum_iP_i^{\mathrm{fix}}$ collects the same fixed communication/circuit loads that enter the physical battery dynamics. The spacecraft bus load remains in the battery dynamics through $P_i^{\mathrm{bus}}(t)$ but is excluded from \eqref{eq:ee}; thus $\Xi(t)$ is a communication-side EE metric rather than total spacecraft EE.

\subsection{Perturbed Battery Queue and Per-Slot Problem}
Let $\Omega_i\in(E_i^{\min},E_i^{\max})$ and $Z_i(t)\triangleq\Omega_i-E_i(t)$. Let $L(t)=\frac12\sum_iZ_i^2(t)$ and define the one-slot Lyapunov change as $\delta_L(t)\triangleq L(t+1)-L(t)$. From \eqref{eq:battrec}, battery saturation can only reduce the squared magnitude of the unsaturated queue update; hence, for bounded harvested/load powers,
\begin{align}
\mathbb E[\delta_L(t)\mid Z(t)]\leq B
+\Delta\sum_i Z_i(t)\,&\mathbb E[ P_i^{\mathrm{load}}(t)\notag\\
&-\theta_iH_i(t)\mid Z(t)].
\label{eq:driftbound}
\end{align}
for a finite action-independent $B$. Let the required per-user QoS floors satisfy $R_j^{\min}\geq0$. Dropping the action-independent harvest term from the drift-minus-utility bound gives
\begin{subequations}\label{eq:p1}
\begin{align}
\max_{p_c,\{p_j\},\{t_j\},r_c}\quad
& V\,\Xi(t)-\Delta\sum_i Z_i(t)P_i^{\mathrm{load}}(t)
\label{eq:p1_obj}\\
\mathrm{s.t.}\quad
& R_j(t)\geq R_j^{\min},\quad \forall j,
\label{eq:p1_qos}\\
& 0\leq e_i(t)\leq e_i^{\max}(t),\quad \forall i,
\label{eq:p1_energy}\\
& \sum_j t_j(t)\leq r_c(t),
\label{eq:p1_common}\\
& r_c(t)\leq \log_2\!\bigl(1+\Gamma_{c,k}(t)\bigr),\quad \forall k,
\label{eq:p1_decode}\\
& p_c(t),\,p_j(t),\,t_j(t)\geq0.
\label{eq:p1_nonneg}
\end{align}
\end{subequations}
where $e_i^{\max}(t)$ is given by \eqref{eq:emax}. The objective in~\eqref{eq:p1_obj} balances the instantaneous
communication-side EE against the battery-state-dependent cost of energy
consumption through the perturbed queues $Z_i(t)$. Constraint
\eqref{eq:p1_qos} guarantees the minimum QoS requirement of each UE,
whereas \eqref{eq:p1_energy} enforces the battery-dependent
communication-power limit at each serving satellite. Constraint
\eqref{eq:p1_common} ensures that the aggregate common-rate allocation
does not exceed the available common-stream rate, while
\eqref{eq:p1_decode} guarantees that the common stream is decodable by
every UE in the serving cluster. Finally, \eqref{eq:p1_nonneg} imposes
non-negativity on all power and common-rate allocation variables.

\subsection{Fixed-Auxiliary Convex Reformulation}
Problem~\eqref{eq:p1} is nonconvex due to the coupled SINR
ratios in the common/private RSMA rates and the fractional
communication-side EE objective. To obtain a tractable per-slot
solution, we first reformulate the rate expressions using the
Lagrangian-dual and quadratic transforms, and then apply a second
quadratic transform to the EE ratio. This yields a convex
fixed-auxiliary resource-allocation subproblem. For a generic ratio $S/I$ with $S\geq0$ and $I>0$, the
Lagrangian-dual transform followed by the scalar quadratic transform
\cite{shen2018fractional} gives the exact variational representation
\begin{align}
\log_2\!\left(1+\frac{S}{I}\right)
&=\max_{\gamma\geq0,\,y\geq0}
\log_2(1+\gamma)-\frac{\gamma}{\ln2}\notag\\
&+\frac{1+\gamma}{\ln2}\left[2y\sqrt{S}-y^2(S+I)\right],
\label{eq:generic_transform}
\end{align}
whose maximizers are $\gamma^\star=S/I$ and $y^\star=\sqrt{S}/(S+I)$. Thus, for the common stream of UE $j$, define
\begin{align}
\hat R_{c,j}
=&\log_2(1+\gamma_{c,j})-\frac{\gamma_{c,j}}{\ln2}\notag\\
&+\frac{1+\gamma_{c,j}}{\ln2}
\Big[2y_{c,j}\sqrt{S_{c,j}}
-y_{c,j}^2(S_{c,j}+I_{c,j})\Big].
\label{eq:surrogate}
\end{align}
where $S_{c,j}=p_ch_j$ and $I_{c,j}=h_j\sum_kp_k+\sigma^2$. The private-stream expression is analogous, with $S_{p,j}=p_jh_j$ and $I_{p,j}=h_j\sum_{k\neq j}p_k+\sigma^2$. For arbitrary auxiliaries these expressions lower-bound the exact rates, with equality and first-order tightness at the closed-form maximizers. For fixed auxiliaries they are concave in $p$, so the common-rate constraints use the convex hypographs
\begin{equation}
r_c\leq \hat R_{c,k}(p,\gamma_{c,k},y_{c,k}),\qquad \forall k.
\label{eq:common_hypo}
\end{equation}
and the QoS constraints use $t_j+\hat R_{p,j}\geq R_j^{\min}$. Since $R_j^{\min}\geq0$, every transformed feasible point has $\hat N\triangleq W\sum_j(t_j+\hat R_{p,j})\geq0$. Let $D=P_0+\sum_ie_i$. A second quadratic transform gives
\begin{equation}
\frac{\hat N(p,t)}{D(p)}
=\max_{\nu\geq0}\;2\nu\sqrt{\hat N(p,t)}-\nu^2D(p),
\label{eq:ee_transform}
\end{equation}
with $\nu^\star=\sqrt{\hat N}/D$ whenever $\hat N\geq0$.

\begin{proposition}\label{prop:convex}
Fix $\{\gamma_{c,j},y_{c,j},\gamma_{p,j},y_{p,j}\}$ and $\nu$. Maximizing
\begin{equation}
V\left[2\nu\sqrt{\hat N(p,t)}-\nu^2D(p)\right]
-\Delta\sum_i Z_i(t)P_i^{\mathrm{load}}(p)
\label{eq:fixed_aux_obj}
\end{equation}
subject to \eqref{eq:common_hypo}, $t_j+\hat R_{p,j}\geq R_j^{\min}$, \eqref{eq:emax}, $\sum_jt_j\leq r_c$, and nonnegativity is a convex optimization (equivalently, maximization of a concave objective over a convex feasible set).
\end{proposition}
\begin{proof}
For fixed auxiliaries, the transformed rates and $\sqrt{\hat N}$ are concave, whereas $D$ and $P_i^{\mathrm{load}}$ are affine; the rate and energy constraints define a convex feasible set.
\end{proof}

The BA-RSMA update repeats three steps within each slot: 1) evaluate the exact SINRs at the current power vector and set the common/private rate auxiliaries to their closed-form maximizers in \eqref{eq:generic_transform}; 2) set $\nu=\sqrt{\hat N}/D$; and 3) solve the fixed-auxiliary convex problem of Proposition~\ref{prop:convex}. Iteration stops when the relative change of the exact per-slot objective in \eqref{eq:p1} falls below a prescribed tolerance.

\begin{theorem}\label{thm:conv}
Fix the cluster, current battery state, and queue weights $\{Z_i(t)\}$. Suppose every fixed-auxiliary subproblem is solved exactly and the feasible set is nonempty and compact. Then the sequence of exact per-slot objective values in \eqref{eq:p1} generated by the alternating updates is non-decreasing and convergent. Under the usual constraint qualification, every accumulation point is a stationary/KKT point of the per-slot problem.
\end{theorem}
\begin{proof}
With $F$ the exact objective and $\widetilde F$ its transformed lower bound, tight auxiliary updates and exact inner maximization give $F(x^{(r)})=\widetilde F(x^{(r)},a^{(r)})\leq\widetilde F(x^{(r+1)},a^{(r)})\leq F(x^{(r+1)})$. Boundedness gives convergence; first-order tightness and the stated constraint qualification yield KKT stationarity.
\end{proof}

\begin{remark}\label{rem:V}
Because $Z_i(t)$ is induced by the bounded physical battery, no unbounded virtual-queue $O(V)$ backlog claim is made. Here $V$ simply controls the per-slot EE--reserve weighting in \eqref{eq:p1}.
\end{remark}

\section{Numerical Results}

\subsection{Simulation Setup and Solver Evaluation}
We consider a $K=3$-satellite cluster serving $U=4$ UEs over $T=96$ slots in a 600-km-class LEO scenario. The carrier frequency and bandwidth are $20$~GHz and $20$~MHz, respectively; the system temperature is $290$~K, receiver noise figure is $5$~dB, and the satellite/UE antenna gains are $35/10$~dBi. The Rician factor is $10$~dB, giving $\rho=0.909$, and the nominal slant ranges vary smoothly between approximately $600$ and $1100$~km with slow $0.6$-dB shadowing.

\begin{figure}[t]
\centering
\begin{subfigure}{\columnwidth}
    \centering
    \includegraphics[width=\linewidth]{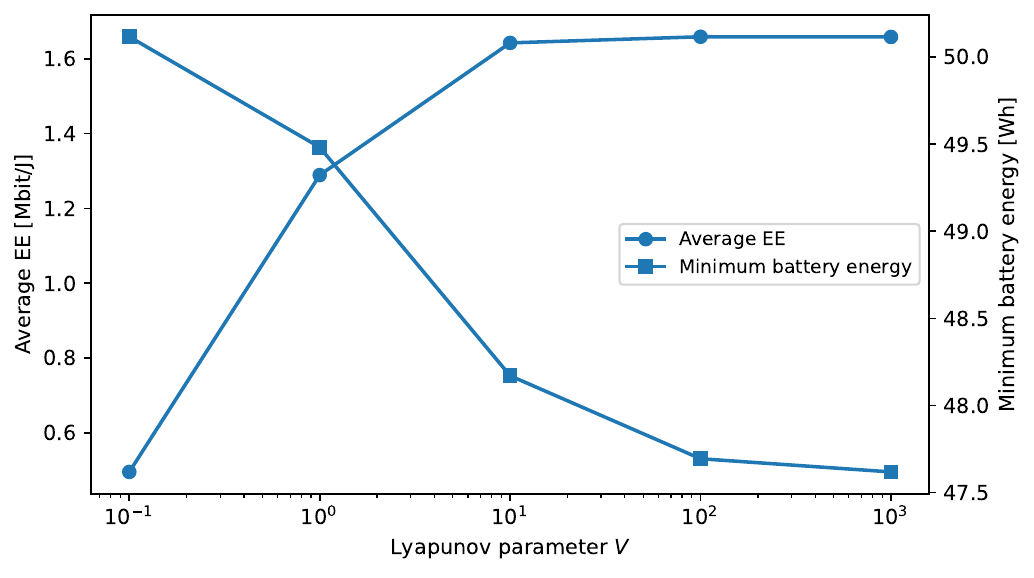}
    \caption{}
    \label{fig:eevsV_a}
\end{subfigure}


\begin{subfigure}{\columnwidth}
    \centering
    \includegraphics[width=\linewidth]{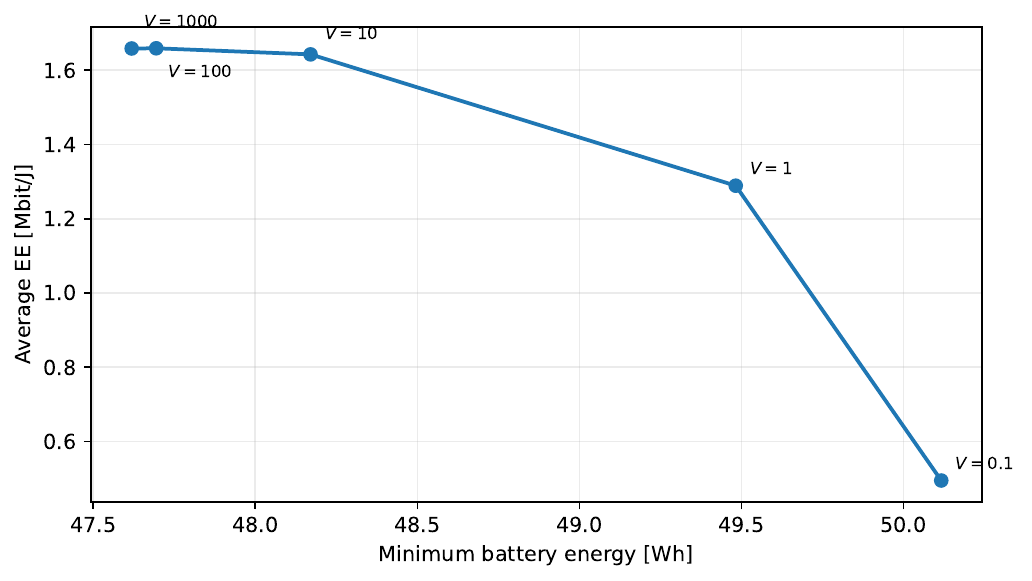}
    \caption{}
    \label{fig:eevsV_b}
\end{subfigure}
\caption{Effect of the Lyapunov parameter on BA-RSMA: (a) average communication-side EE and minimum battery energy versus $V$; (b) corresponding EE--battery-reserve operating points.}
\label{fig:eevsV}
\end{figure}

The fixed power shares are computed offline from the long-term design trajectory using \eqref{eq:power_share}, yielding approximately $\boldsymbol\xi=[0.323,0.333,0.344]$. We set $\eta_i=0.35$, $P_i^{\rm bus}=20$~W, $P_i^{\rm fix}=5$~W, $\theta_i=0.9$, $\Delta=30$~s $=1/120$~h in the battery updates, $E_i^{\max}/E_i^{\min}=100/20$~Wh, $\Omega_i=68$~Wh, $R_j^{\min}=0.1$~bit/s/Hz, a 35\% eclipse fraction, and bounded sunlit harvest uncertainty $\delta_{H,i}=5$~W. The nominal $V$-sweep starts from $E_i(0)=50$~Wh with $\widehat H_i=60$~W in sunlit slots; the battery-stressed comparison uses $E_i(0)=28$~Wh and $\widehat H_i=40$~W.

For every slot, \eqref{eq:emax} and the fixed power share $\xi_i$ imply
$p_c(t)+\sum_jp_j(t)\leq\min_{i\in\mathcal S}\{\eta_i e_i^{\max}(t)/\xi_i\}$. Thus, the tightest satellite limits the implemented power rather than the sum of the individual energy headrooms. The realized harvest obeys the bounded-error model of Assumption~1, and the true battery recursion uses \eqref{eq:battrec}; no lower-bound battery clipping is applied.

The BA-RSMA long-horizon curves are generated by the proposed dual-plus-quadratic-transform alternating algorithm, warm-started from the previous slot to reduce the per-slot iteration count. For these long-horizon campaigns, the warm-started solver uses a relative outer stopping tolerance of $10^{-5}$ and a 60-iteration safeguard. QoS and battery constraints are rechecked before accepting each solution. For numerical cross-validation, the original per-slot problem \eqref{eq:p1} is also solved directly by SLSQP on 12 representative stressed slot--state instances. With a $10^{-7}$ relative stopping tolerance and a 100-iteration safeguard in this validation, all 12 proposed-algorithm runs are monotonic and feasible, the maximum final constraint residual is $9.3\times10^{-9}$, and the maximum relative objective gap to the direct reference is $2.8\times10^{-7}$. The median and maximum outer iteration counts are 65 and 78, respectively; Fig.~\ref{fig:convergence} shows a representative trajectory. Myopic-RSMA uses the same continuous formulation with the sole change $Z_i(t)=0$ and is solved by the direct nonlinear-programming reference.

\subsection{Benchmark Models}
The primary BA-RSMA results are generated by the proposed alternating algorithm for \eqref{eq:p1}. Continuous Myopic-RSMA uses the same feasible set with $Z_i(t)=0$ and the direct nonlinear-programming reference, isolating the queue-aware battery term at the formulation level; the remaining schemes are reduced-dimensional benchmarks.

\emph{Reduced-order BA-RSMA:} For $(P,\beta_c)$, $p_c=\beta_cP$; the remaining $(1-\beta_c)P$ is split as $p_j=(1-\beta_c)P h_j^{-\alpha_p}/\sum_kh_k^{-\alpha_p}$, with $\beta_c$ searched over ten values in $[0.05,0.8]$ and $\alpha_p\in\{0,0.5,1,1.5,2,3\}$.

\emph{NOMA:} Users are ranked by channel gain (rank 1 strongest); user $k$ cancels all weaker-ranked signals via SIC before decoding its own, $\mathrm{SINR}_k=p_kh_{(k)}/(h_{(k)}\sum_{i<k}p_i+\sigma^2)$, with $p_k=P h_{(k)}^{-\alpha}/\sum_ih_{(i)}^{-\alpha}$ and the same exponent family.

\emph{OMA:} Equal resource shares $\tau_j=1/U$ are used with $p_j=P h_j^{-\alpha_o}/\sum_kh_k^{-\alpha_o}$ and $R_j=\tau_j\log_2(1+p_jh_j/(\tau_j\sigma^2))$.
For every scheme, $R_j(t)\geq R_j^{\min}$ and the battery-dependent power cap are checked at every candidate point; infeasible slots are counted as QoS outages.

\subsection{Results Discussion}
Fig.~\ref{fig:eevsV} quantifies the control role of $V$ for the BA-RSMA formulation using the proposed alternating implementation. In Fig.~\ref{fig:eevsV_a}, increasing $V$ from $0.1$ to $100$ raises the average EE from $0.495$ to $1.659$~Mbit/J while reducing the minimum battery energy from $50.12$ to $47.69$~Wh; at $V=1000$, the EE remains essentially saturated at $1.659$~Mbit/J while the reserve falls slightly further to $47.62$~Wh. The average radiated power increases from $0.41$ to about $5.01$~W. Fig.~\ref{fig:eevsV_b} exposes the corresponding EE--reserve operating frontier: the move from $V=0.1$ to $V=10$ produces a large EE gain ($0.495$ to $1.643$~Mbit/J) for a moderate reserve reduction ($50.12$ to $48.17$~Wh), whereas the frontier nearly flattens by $V\approx100$. Hence, increasing $V$ beyond this point spends additional reserve for negligible EE improvement. No QoS outage or battery-floor violation occurs in the sweep, consistent with Remark~\ref{rem:V}.
\begin{figure}[t]
\centering

\begin{subfigure}{\columnwidth}
    \centering
    \includegraphics[width=\linewidth]{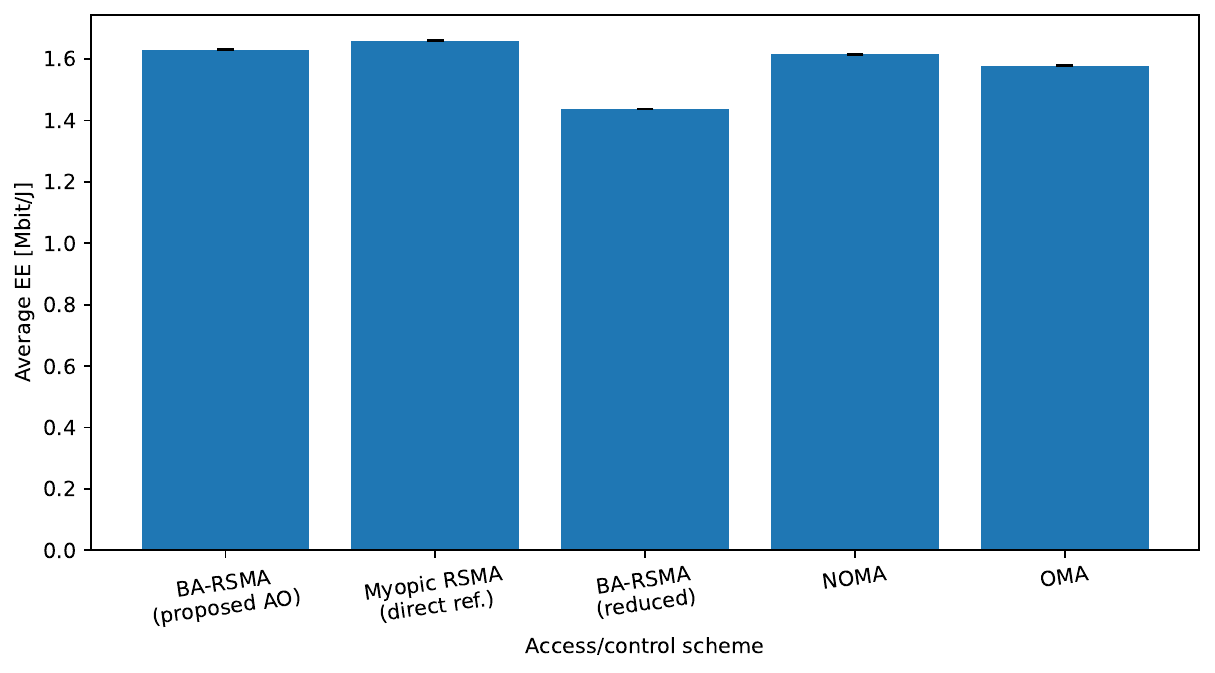}
    \caption{}
    \label{fig:policy_a}
\end{subfigure}


\begin{subfigure}{\columnwidth}
    \centering
    \includegraphics[width=\linewidth]{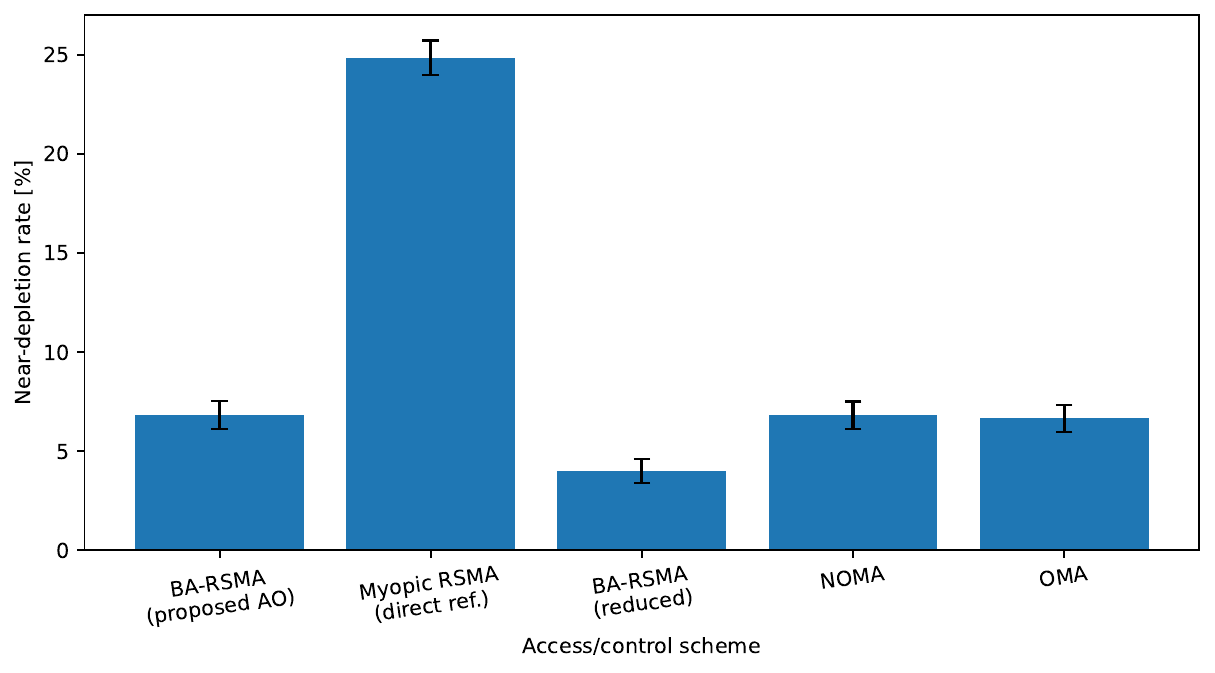}
    \caption{}
    \label{fig:policy_b}
\end{subfigure}
\caption{Battery-stressed comparison:
(a) average EE and (b) near-depletion rate
($E_i(0)=28$~Wh, $\widehat H_i=40$~W, $V=20$).}
\label{fig:policy}
\end{figure}

Fig.~\ref{fig:policy} reports a 20-seed battery-stressed Monte Carlo experiment with lower initial charge and weaker solar harvesting. Near depletion is defined as
\begin{equation}
R_{\rm dep}=\frac{1}{T}\sum_{t=0}^{T-1}
\mathbbm{1}\!\left\{\min_i E_i(t)\leq E_i^{\min}+0.1E_i^{\min}\right\}.
\end{equation}
All schemes remain QoS-feasible in all tested slots. The proposed-algorithm BA-RSMA solution attains $1.631$~Mbit/J with a mean near-depletion rate of $6.82\%$. The continuous Myopic-RSMA ablation, which removes the battery-queue term and is solved by the direct nonlinear-programming reference, attains the highest raw EE, $1.659$~Mbit/J, but increases near depletion to $24.84\%$. Thus, battery-aware control incurs only about a $1.7\%$ EE penalty while reducing near-depletion exposure by $18.02$ percentage points (about $72.5\%$ relative reduction). BA-RSMA also maintains a higher mean minimum battery state, $21.68$~Wh versus $20.99$~Wh.

The low-complexity reduced-order BA-RSMA achieves $1.437$~Mbit/J and the lowest near-depletion rate, $4.01\%$, illustrating a more conservative operating point within the RSMA family. NOMA and OMA achieve $1.616$ and $1.579$~Mbit/J, respectively, with mean near-depletion rates of $6.82\%$ and $6.67\%$. Hence, under the considered benchmark implementations, the proposed-algorithm BA-RSMA exceeds NOMA and OMA in EE by about $0.9\%$ and $3.3\%$, respectively, while providing essentially the same near-depletion exposure. This ordering is not claimed as a universal multiple-access ranking because NOMA and OMA use reduced-dimensional benchmark families.

\begin{figure}[t]
\centering
\includegraphics[width=\linewidth]{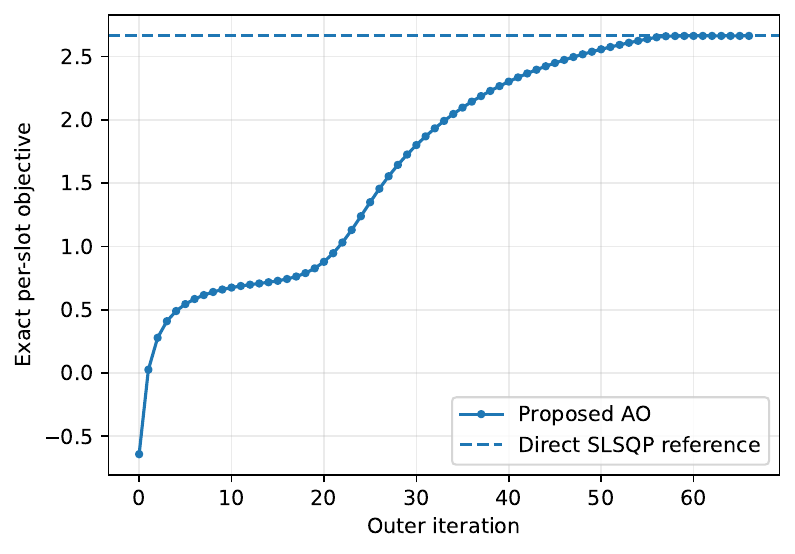}
\caption{Proposed alternating-solver convergence for a representative stressed slot; the dashed line denotes the direct SLSQP reference solution of \eqref{eq:p1}.}
\label{fig:convergence}
\end{figure}

Fig.~\ref{fig:convergence} illustrates the proposed dual-plus-quadratic-transform implementation for a representative stressed slot. The exact objective increases monotonically and approaches the direct SLSQP reference solution. Across the 12 cross-validation instances, every run is monotonic and feasible, with maximum relative objective gap $2.8\times10^{-7}$ and maximum residual $9.3\times10^{-9}$, numerically supporting Theorem~1.

\section{Conclusion}

This letter developed battery-aware one-layer RSMA for solar-powered cell-free LEO downlinks. A perturbed physical-battery queue couples RSMA power allocation to stored energy, while robust energy causality and harvest curtailment guarantee battery safety. The dual-plus-quadratic transform yields convex fixed-auxiliary subproblems with monotonic alternating updates and matches the direct nonlinear-programming reference to within $2.8\times10^{-7}$ in relative objective over 12 tested states. In the 20-seed stressed experiment, the proposed-algorithm BA-RSMA sacrifices only $1.7\%$ EE relative to continuous Myopic-RSMA while reducing near-depletion from $24.84\%$ to $6.82\%$. Under the considered benchmarks, it also exceeds NOMA and OMA in EE. Future work will consider dynamic clustering and geometry-based multi-layer RSMA.

\bibliographystyle{IEEEtran}
\bibliography{References}

\end{document}